\documentclass[11pt]{article}
\usepackage{antonio}

\title{Proximity Gaps Conjecture Fails Near Capacity over Prime Fields}
\author{Antonio Kambiré}
\date{April 1st 2026}

\begin{document}

\maketitle

\begin{abstract}
In this report we flesh out a sketch by Krachun and Kazanin to prove that for a certain family of Reed--Solomon codes, proximity gaps fail at radii that are $O(1/\log n)$ below the capacity rate of the code, where $n$ is the length of the code. 
\end{abstract}

\section{Introduction}

The proximity gaps conjecture, introduced in \cite{BCISS20}, asserts that if many points on an affine line \(f+zg\) are each close to a Reed-Solomon code, then the line itself must be explained by a nearby codeword pair, that is, the pair \([f,g]\) is close to the corresponding interleaved Reed-Solomon code, a condition called correlated agreement. While this phenomenon is well-understood up to the Johnson bound in settings covered in \cite{BCISS20} and subsequent papers, its behavior above the Johnson bound is still less clear and is the subject of active research. 

We give a detailed and self-contained proof of a family of Reed--Solomon codes $\mathrm{RS}[\mathbb{F}_p,\mathcal D,k]$ for which proximity gaps fail at radii that are $O(1/\log n)$ below the capacity value $1-k/n$, where $n = \abs{\mathcal D}$. This construction builds on a sketch by Krachun and Kazanin~\cite{KK26}, making all bounds and parameter choices explicit. In concurrent work, Appendix A.5 of \cite{DLULIS26} used this result to formalize a conjecture on list- and curve-decodability properties of Reed-Solomon codes over prime fields, up to the information-theoretic limit.

Our main theorem is the following. We construct block lengths $n$ and dimensions $k$, a prime field $\mathbb{F}_p$, and words $f,g\in\mathbb{F}_p^{\mathcal D}$ such that for $\delta = (1-\frac{k}{n})-\Omega(1/\log n)$ there are at least $n^C$ distinct scalars $z\in\mathbb{F}_p$ with $\Delta(f+zg,\mathcal C)\le \delta$, while simultaneously $\Delta([f,g],\mathcal C^2)>\delta$, where $\mathcal C=\mathrm{RS}[\mathbb{F}_p,\mathcal D,k]$. The construction has two components. First, a coding-theoretic argument produces an explicit line with many near-codewords, following the same multiplicative-subgroup and sumset template as Theorem~7.1 of \cite{BCHKS25}. Second, a quantitative number-theoretic argument based on a quantitative Linnik theorem produces primes $p\equiv 1 \pmod n$ with $p< n^{A}$ for which the relevant subset-sum values are distinct in $\mathbb{F}_p$, ensuring polynomially many distinct parameters $z$.

\section{Statement of the Result and Proof}

We prove the following theorem which shows the failure of proximity gaps at vanishing distance near capacity:

\begin{proofbox}
    \begin{theorem}  \label{thm:main}
        For every constant $C>0$ and rate $\rho \in (0,\frac{1}{2})$, there exist infinitely many block lengths $n$, dimensions~$k$, such that with
        $\delta = (1-\tfrac{k}{n}) - \Omega(\tfrac{1}{\log n})$,
        the following hold:
        
        \begin{itemize}
            \item There exists a prime $p <n^{A}$ with $p\equiv 1 \pmod n$, for some constant $A=A(\rho,C)$.
            \item Let $\omega$ be a primitive $n^{th}$ root of unity in $\mathbb{F}_p^\times$, set $\mathcal D=\langle \omega\rangle$, and $\mathcal C=\mathrm{RS}[\mathbb{F}_p,\mathcal D,k]$. \\
            Then there exist $f,g\in \mathbb{F}_p^{\mathcal D}$ such that
            \[
            \Bigl|\bigl\{ z\in\mathbb{F}_p : \Delta(f+zg,\ \mathcal C)\le \delta \bigr\}\Bigr|
            \geq n^C,
            \qquad
            \Delta([f,g],\mathcal C^2)>\delta.
            \]
        \end{itemize}
    \end{theorem}
\end{proofbox}

\noindent\textbf{Setting parameters.} 
Before we prove the theorem, let us first explain what is the constant in the $\Omega(\tfrac{1}{\log n})$ expression in the definition of $\delta$.
To do so we choose the following parameters in $\mathbb{Z}_{>0}$:
    \begin{itemize}
        \item $C > 0$, controlling how many different sums we get
        \item $\rho := \frac{u}{2^v} \in (0,\frac{1}{2})$ for some $u,v \in \mathbb{Z}_{\geq 0}$ with $u <2^{v-1}$, chosen to be the rate of the code
        \item Let $L(\rho,C) = \max\{\frac{C}{\rho\log(1/(2\rho))}, \frac{9}{2\log 8}\}$ and pick $K=K(\rho,C)$ a power of $2$ between $L(\rho,C)$ and $2L(\rho,C)$. Such $K$ always exists -- we could pick $K=2^{\lfloor \log_2 L(\rho,C)\rfloor+1}.$
        \item $s :=2^\alpha$ for some positive integer $\alpha$ large enough to satisfy:
        \begin{itemize}
            \item $\alpha\geq v$ to guarantee that $r$ below is an integer
            \item 
            \(
                \begin{cases}
                    \alpha \geq \log_2 K \\
                    K \leq \frac{2^\alpha}{\alpha}
                \end{cases}
            \) ensuring that $K \mid 2^\alpha$ and $2^\alpha/K \geq \alpha$ so that $m$ below is a power of $2$

        \end{itemize}
        This controls how close we get to capacity
        \item $r: = \rho s+2 = u2^{\alpha-v}+2$ is a positive integer controlling the relative distance
    \end{itemize}
    
    Now, for $m: = 2^{2^{\alpha}/K -\alpha}$ a power of $2$, we choose:
    \[
        n := sm, \qquad k := (r-2)m.
    \]
    We obtain the following identities:
    \begin{align*}
    (1)\quad & \rho 
            = \frac{r-2}{s} 
            = \frac{(r-2)m}{sm} 
            = \frac{k}{n} \\
    (2)\quad & K\log_2 n 
            = K\log_2(sm) \\
            &\phantom{K\log_2 n}
            = K\log_2(2^\alpha 2^{2^{\alpha}/K  -\alpha}) \\
            &\phantom{K\log_2 n} = K 2^{\alpha}/K\\
            &\phantom{K\log_2 n}= 2^\alpha = s
    \end{align*}
    Then we set $\delta \coloneq 1-\frac{r}{s}$ which is $\eta \coloneq (1-\rho)-\delta = \frac{2}{s} \in \Omega\left(\frac{1}{\log(n)}\right)$ away from capacity.  
    
    With these parameters in place we are ready to prove Theorem~\ref{thm:main}. 
    \medskip

\begin{proof}[Proof of Theorem~\ref{thm:main}]  
    The first part of this proof is a coding theory argument on the construction of the counterexample and the second part is a more combinatorial and number theoretic one which allows us to prove the quantitative and existential aspects of the theorem.

\noindent \textbf{Constructing the Counterexample.} We work with the Reed-Solomon code $\mathcal{C} = \mathrm{RS}[\mathbb{F}_p,\mathcal{D},k]$.
Let $\xi$ be a primitive $s^{th}$ root of unity in $\mathbb{F}_p^\times$ and denote $H:=\langle \xi \rangle \subset \mathcal{D}$.
We define the set of the different sums that we can get by adding exactly $r$ elements from $H$ as follows:
\[
H^{(+r)} = \left\{ \sum_{i=1}^{r} e_i \;\middle|\; 
e_1,\dots,e_r \in H \text{ are distinct} \right\}.
\]

Now, for $f:= X^{rm}$ and $g = X^{(r-1)m}$, consider the following line in $\mathbb{F}_p^D$:
\[
L := \{f + \lambda \cdot g \mid \lambda \in \mathbb{F}_p\} \subset \mathbb{F}_p^D.
\]

Given some $\lambda =\xi_1+\xi_2+\dots+\xi_r  \in H^{(+r)}$, we claim that:
\[
\Delta(X^{rm}+\lambda\cdot X^{(r-1)m}, \mathcal{C}) \leq \delta.
\]

To show this, first notice that the number of cosets of $H$ in $D$ is $\frac{|D|}{|H|} = \frac{n}{s}=m$. Then, we can pick $r$ of these cosets:
\[
\text{for }j = 1,2,\dots,r,\qquad H_{j} := \{a \in D \mid a^m = \xi_j\}.
\]

We get the given identity by analyzing the roots of $X^m-\xi_j$:
\begin{align*}
    \prod_{a \in H_1 \cup \cdots \cup H_r} (X - a)
&=
\prod_{j=1}^r \left(\prod_{a :\, a^m = \xi_j} (X - a)\right)
=
\prod_{j=1}^r (X^m - \xi_j) \\
&=
X^{rm} - (\xi_1 + \cdots + \xi_r) X^{(r-1)m} + R(X) \\
&= X^{rm} - \lambda X^{(r-1)m} + R(x)
\end{align*}
where $R(x)$ is some polynomial of degree at most $(r - 2)m)$.

From this identity, we can easily see that the expression $X^{rm} - \lambda X^{(r-1)m} + R(x)$ vanishes whenever $a \in H_1 \cup \cdots \cup H_r$. This means that $X^{rm} - \lambda X^{(r-1)m}$ agrees with some polynomial of degree at most $(r - 2)m$, namely $R(X)$, on $H_1 \cup \cdots \cup H_r$, which is a set of size $rm = (1 - \delta)n$. Therefore, the Hamming distance to the code must be less than or equal to $\delta$.

On the other hand, correlated agreement does not occur at this Hamming distance. For the sake of contradiction, assume it did. Then, we would have some $D' \subset D$ of size $(1 - \delta)n$ for which every point on $L$ agrees with some polynomial of degree at most $k$. In particular, $X^{(r-1)m}$ agrees with a polynomial $q(X) \in \mathbb{F}_p[X]$ of degree at most $k$ on $D'$, but this implies that $|D'| \leq k = (r-2)m$ as $q(X)$ can have at most $k$ roots in $\mathbb{F}_p$. This contradicts $|D'|=rm $.

\noindent\textbf{Counting the Number of Sums.} Now we prove that $|H^{(+r)}|$ is big enough to satisfy the quantitative requirements of the theorem. For this, we make use of the quantitative version of Linnik's Theorem which we restate for convenience:

\begin{proofbox}
If $n$ is sufficiently large and \(x \geq n^L\) with
\[
L = \max\curl{4c_2, \frac{4}{c_1}\log 8c, \frac{4}{c_3}\frac{\log 8c}{|\log c_1|}} <3
\]
where $c_1,c_2$, and $c_3$ are absolute constants, then:
\[
\displaystyle
\psi(x;n,a)
:= \sum_{\substack{k \,\leq \, x \\  k\, \equiv\, a \pmod n}} \Lambda(k)
\geq \frac{x}{\vphi(n)\sqrt{n}}.
\]
\end{proofbox}

For the following, let $\log \equiv \ln$. Using the version of Linnik's theorem stated above, we have that:
\[
\displaystyle
\theta(x;n,a)
:= \sum_{\substack{p \leq x \text{ prime} \\ p \equiv a \pmod n}} \log(p)
\geq \frac{x}{\vphi(n)\sqrt{n}}+O(\sqrt{x}).
\]

At this point, assume $4^s \geq n^3$ (equivalently $4^{(r-2)n/k}\geq n^3$ which is true if $n$ is large). In the interval $[4^s,8^s]$, let
\(
T := \#\{\,p\in[4^s,8^s]\text{ prime} : p \equiv 1 \!\!\pmod n\,\}.
\)
Then:
\begin{align*}
\theta(8^s;n,1) - \theta(4^s;n,1)
&= \sum_{\substack{4^s \leq p \leq 8^s \\ p \equiv 1 \bmod n}} \log(p)
\le T \log(8^s), \\
\implies T \;&\ge\; \frac{\theta(8^s;n,1)-\theta(4^s;n,1)}{\log(8^s)}.
\end{align*}
Using the lower bound $\theta(8^s;n,1)\geq \dfrac{8^s}{\varphi(n)\sqrt{n}}$ and the trivial upper bound
\[
\theta(4^s;n,1)
=\sum_{\substack{p\le 4^s\\ p\equiv 1\!\!\pmod n}}\log(p)
\le \left(\frac{4^s}{n}+1\right)\log(4^s),
\]
we get
\[
T \;\ge\; \frac{1}{\log(8^s)}
\left(\frac{8^s}{\varphi(n)\sqrt{n}}-\left(\frac{4^s}{n}+1\right)\log(4^s)\right).
\]
Since $4^s \geq n^3$, for $n$ sufficiently large we have
\[
\left(\frac{4^s}{n}+1\right)\log(4^s)\le \frac12\cdot \frac{8^s}{\varphi(n)\sqrt{n}},
\]
and therefore
\[
T \;\ge\; \frac{8^s}{2\,\varphi(n)\sqrt{n}\,\log(8^s)}.
\]
Since $n=2^t$, we have $\varphi(n)=n/2$, hence $\varphi(n)\sqrt{n}=n^{3/2}/2$, so it follows that:
\[
T \;\ge\; \frac{8^s}{2\,\varphi(n)\sqrt{n}\,\log(8^s)}
\;=\; \frac{8^s}{n^{3/2}\log(8^s)}.
\]

Then, using $s = K\log n$, we have:
\[
n = e^{s/K}\qquad\text{and}\qquad \log(8^s)= s\log 8.
\]
Substituting these into the bound $T \geq \dfrac{8^s}{n^{3/2}\log(8^s)}$ gives us:
\[
T \;\ge\; \frac{8^s}{e^{\frac{3s}{2K}}\,(s\log 8)}.
\]
Since $e^{\frac{3s}{2K}} = 8^{\frac{3s}{2K\log 8}}$, we obtain:
\[
T \;\ge\; \frac{1}{s\log 8}\,8^{\,s\left(1-\frac{3}{2K\log 8}\right)}.
\]

Let $P(x):=\Phi_s(x)$ and $Q(x):=x^{i_1}+\dots+x^{i_r}-(x^{j_1}+\dots+x^{j_r})$.
We are now interested in the number of ``bad'' primes in this interval: those that divide the resultant of $P$ and $Q$. We have:
\[
|Res(P,Q)|=\left|\prod_{x :P(x)=0}Q(x) \right| \leq (2r)^{s/2} \leq s^s.
\]

If this value vanishes in $\mathbb{F}_p$ for some prime $p$, then the prime allows for different sums to collide. Let $B$ denote the number of such bad primes in $[4^s,8^s]$ for a specific $r$-tuple pair. Then the fact that any prime we pick satisfies $p\geq 4^s$ and that each bad prime divides the resultant implies that:
\[
 (4^s)^B \leq s^s \implies B \leq \frac{\log(s)}{\log(4)} = \log_4(s).
\]

This means that each pair of $r$-tuples creates at most $B$ collisions (generate $B$ bad primes), giving us a total number of at most:
\[
B\binom{s}{r}^2  \leq \log_4(s)(2^s)^2 \ll T
\]
triples $(p,R_1,R_2)$ where $p$ is a bad prime and $R_1,R_2$ are all pairs of $r$-tuples. The last inequality holds for large $s$ given $K >\frac{9}{2\log 8}$. Therefore there must exist a good prime in the interval $[4^s,8^s]$ with $p < 8^s = 8^{K\log n} =n^{K\log 8}$ so we can pick $A=K\log 8$, for which we have a number of different sums given by:
\[
a := |H^{(+r)}| = \binom{s/2}{r} \geq \left(\frac{s}{2r}\right)^r.
\]

Rewriting $r=\rho s+2$, we get:
\begin{align*}
a &\geq \left(\frac{s}{2(\rho s+2)}\right)^r
= \left(\frac{1}{2(\rho+\frac{2}{s})}\right)^r \\
&\approx \left(\frac{1}{2\rho}\right)^{\rho s+2} \\
&= \left(\frac{1}{2\rho}\right)^{\rho K\log n+2} \\
&= n^{\rho K \log(1/(2\rho))} \left(\frac{1}{2\rho}\right)^2.
\end{align*}

Since $K>\frac{C}{\rho\log(1/(2\rho))}$ and $\frac{1}{2\rho}>1$ as defined above, we get $a > n^C$ as needed.
\end{proof}

\bibliographystyle{plain}
\bibliography{bib}

\end{document}